\documentclass[runningheads]{llncs}

\usepackage[T1]{fontenc}
\usepackage{amsmath,amssymb}
\usepackage{mathtools}
\usepackage{enumitem}
\usepackage{booktabs}
\usepackage{microtype}
\usepackage{dsfont}
\usepackage[hidelinks]{hyperref}

\title{A Repeated-Game Framework for Incentives in Decentralized Infrastructure Protocols}
\titlerunning{Repeated-Game Incentives in Decentralized Infrastructure Protocols}

\author{Mustafa Qazi}
\authorrunning{M. Qazi}
\institute{Volt Capital}

\spnewtheorem{assumption}{Assumption}{\bfseries}{\itshape}

\newcommand{\R}{\mathbb{R}}

\begin{document}
\maketitle

\begin{abstract}
We introduce a repeated dynamic incentive framework for characterizing when compliance, or full-effort honest service provision, is incentive compatible in Decentralized Physical Infrastructure Networks (DePIN). We model quality control in this setting as a repeated moral-hazard problem between the protocol and each provider, where compliance is enforced by both slashing posted collateral and the discounted threat of demotion to probation tiers on a reputation ladder. Our main contribution is the "deterrence ratio" $\Gamma$, the worst-case ratio of a deviation's private gain to its marginal probability of detection. We find that if any profitable deviation does not increase the fail probability relative to compliance, then binary public-outcome protocols cannot deter it. When all profitable deviations have positive detection gaps and no weakly costlier action is less likely to fail, compliance is sequentially incentive compatible if and only if immediate slashing plus discounted reputation loss is at least $\Gamma$ at every reputation tier. We use this condition to formulate a protocol design problem mapping service primitives to stake requirements, reward schedules, probation rules, and audit frequency. 

\end{abstract}

\section{Introduction}\label{intro}
Infrastructure systems typically have three major inputs in producing service output: \textit{capital}, or the hardware used to produce the infrastructure service, \textit{operational effort}, or the ongoing maintenance and provision of capital, and \textit{coordination}, the routing, load-balancing, and pricing layers that unify disjointed production nodes. Centralized operators control all of these inputs. Decentralized Physical Infrastructure Networks (DePIN) \cite{lin2024} instead outsource the provider-side inputs to individual \textit{providers} who bring their own capital and provision service in exchange for protocol-specified rewards. This model enables elastic supply without large upfront investment, but the loss of control over production nodes raises the problem at the focus of this work, quality control: how can the protocol guarantee that provider work is (a) actually being done and (b) at or above some threshold of quality? For example, a prevalent issue in DePINs is \textit{self-dealing}, where providers purchase or spoof the purchase of their own service to profit from token emissions. Even more detrimental to the network is the practice of spoofing signals that resemble actual work done without doing any work at all. In practice, disincentivizing or blocking this behavior is highly dependent on the type of service being provisioned: spoofing location data is structurally different from spoofing AI inference workloads. As a result, incentive-compatible mechanisms cannot generally be designed for \textit{all} DePINs. 

Our work takes a different approach in light of this reality: we quantify the difficulty of enforcing full-effort provider service provisioning as the \textit{deterrence ratio}, a quantity that captures the binding deviation and network state as the worst-case scenario. Our work builds on the single-shot model of Milionis et al. \cite{milionis2025}, a more general study of incentive-compatible signal recovery when signals are manipulable. We take their signal-recovery framework as a foundation for the audit layer and study the dynamic enforcement problem that remains once reports are aggregated on the protocol side.

Current DePIN designs use ad hoc mixes of posted collateral slashing, auditors, and provider reputation tracking to address quality control. Our work formalizes this problem as a repeated moral-hazard game between providers and a committed public protocol, deriving necessary and sufficient conditions under which compliance is sequentially incentive compatible. The repeated game framework is essential because it reflects the two enforcement channels available to protocol designers. In a one-shot setting, the only deterrent is immediate punishment, which must be sized to cover the worst-case deviation. In the repeated game, our model can leverage the threat of reduced future rewards. The repeated-game setting also better reflects the decisions and actions of providers in real-world protocols: providers condition their behavior on past outcomes and have different temptations in different network conditions. Providers facing high costs of compliant behaviors in one state may choose to deviate even if compliance is optimal on average. Our state-conditional analysis addresses this concern.

\subsection{Related Work}
Our work is most closely related to that of Milionis et al. \cite{milionis2025}, who introduce the first formal treatment of eliciting signals from self-interested sources that may manipulate signals adversarially. They characterize a \textit{source identifiability} condition, a necessary and sufficient condition for a strictly truthful elicitation mechanism, constructed via a family of mechanisms in the peer-prediction literature \cite{mrz2005}. Our separation condition (Def. \ref{StateConditionalAuditSeparation}) rests on this foundation. Their work is a generalized framework in a single-shot model, from which we take their elicitation guarantees and build a repeated-game framework on top, specific to a DePIN environment. Our contribution begins after their reporting layer, investigating when the resulting monitoring technology is strong enough to support dynamic provider incentives through slashing and reputation. Our detection primitives are the reduced form of a costly-state-falsification monitoring layer (Appendix \ref{app: foundationDetectionPrimitives}) in the tradition of Lacker and Weinberg \cite{lackerweinberg89}, whose test-design branch \cite{perezrichetskreta22} studies how much detection power a test retains under optimal manipulation. The work of \cite{crystalgorenkominers2025} proves incentive-compatibility for the Shelby protocol in a one-shot model. Their results show that on-chain audits and peer audits restore IC over purely off-chain audits. Our framework can be thought of as a generalized form of their results with the added power of reputation. 

Our recursive incentive constraints are closely related to repeated games with imperfect public monitoring \cite{aps1990,flm1994}, in which public signals and continuation values support incentive-compatible behavior. Our protocol, however, commits ex ante to a public automaton that applies mechanical rewards, slashing, and reputation transitions as a function of public outcomes, so we do not characterize the full set of perfect public equilibrium payoffs and instead characterize when compliance is sequentially incentive compatible under a fixed protocol. Classical repeated principal-agent works \cite{radner1985} are conceptually related in studying how noisy public signals trigger future punishments. Other models \cite{greenporter1984} develop the same logic in repeated oligopoly as opposed to our agency setting. 

Our modeling at the stage game level is a repeated moral-hazard problem in the sense of Holmstrom \cite{holmstrom1979}, where effort is unobservable and can only be inferred from noisy signals. Repetition matters because the protocol can discipline current deviations through both immediate punishment and the threat of future demotion in reputation. Our key result, the deterrence ratio $\Gamma$ (\ref{def:deterrence-ratios}) isolates the \textit{state-conditional} cost saving from deviation against the marginal detection ability of the audit. The resulting incentive inequality is a dynamic analogue of the classical deterrence condition of Becker \cite{becker68,polinskyshavell00}, with the sanction endogenized as slashing plus continuation-value loss and the gain and detection probability endogenized through adversarial manipulation. Our use of a reputation channel is most analogous to the no-shirking condition of Shapiro and Stiglitz \cite{shapirostiglitz1984}, in which the threat of losing a rent-bearing position deters shirking under imperfect monitoring. We generate rent through a reputation ladder with stake slashing-backed demotion to lower tiers. 

While our work is theoretical, it endogenizes primitives shared by deployed protocols: Filecoin \cite{filecoin2017} slashes pledged collateral on failed storage proofs, and Helium \cite{helium2018} verifies claimed wireless coverage through mutual hotspot witnessing. Neither whitepaper derives conditions for sustained compliant behavior, and both networks have suffered substantial location-spoofing and self-dealing attacks, which our framework targets and quantifies.

\section{Model}\label{model}

\subsection{Players and Information Structure}\label{PlayersAndInformation}
Time is discrete, $t = 1, 2, \dots$, with common discount factor $\delta \in (0, 1)$. The strategic players are a finite set of \textit{providers} indexed $i \in \mathcal{N} = \{1, \dots, N \}$. The protocol is a committed public automaton that maps public outcomes to rewards, slashing, and reputation transitions. We assume that these parameters are fixed by the protocol designer and do not change strategically during the game. We define the collection of stake requirements, slashing rules, rewards, reputation transition rules, and lockup rules as a committed protocol $\Pi = (S, \lambda, r, T, \Lambda_{\text{lock}})$, a non-strategic automaton whose components are introduced formally below. We fix $\Pi$ for Sections 2--4 and study the designer's choice of $\Pi$ in Section 5. 

Auditors are part of the monitoring technology. They observe provider signals and submit reports, assumed truthful as generated by an incentive-compatible recovery mechanism similar to Milionis et al. \cite{milionis2025}. Each period, the protocol's audit process yields a binary public outcome $z_t^i \in \mathcal{Z} = \{0, 1\}$, where $z=1$ denotes a fail and $z=0$ denotes a pass.

In every period, provider $i$ has a privately observed state $x_t^i$ drawn independently from a distribution $\mu$ over a finite set $\mathcal{X} = \{x_1, x_2, \dots, x_M\}$ . The state captures exogenous conditions outside of the provider's control like network health, hardware health, service demand, and other similar factors. 

Let the provider action space be a finite set $\mathcal{A}$. We denote $a^* \in \mathcal{A}$ as \textit{compliance}, or honest full-effort service provision. All other elements of $\mathcal{A}$ represent deviations of any kind: shirking, self-dealing, signal manipulation, etc. Conditional on $x_t^i$, provider $i$ chooses an action $a_t^i \in \mathcal{A}$.

\subsection{Public State and Player Histories}\label{PublicStateHistories}
A convenient feature of repeated games on blockchains is that punishments do not have to be sequentially rational. Whereas the punishing principal in classical repeated principal-agent games \cite{radner1985} must find it optimal to carry out punishments, smart contracts execute slashing and reputation updates mechanically as protocol code conditional on public outcomes, eliminating the credibility problem entirely. 

\subsubsection{Stake Requirements}\label{StakeRequirements}

Upon entry into the protocol, each provider must post a \textit{minimum stake collateral} $\underline{S} \in [0, \bar{S}]$, where $\bar{S} <\infty$ is the maximum stake\footnote{Our analysis is agnostic to the choice of numeraire for stake collateral.}. The stake is locked for $\Lambda_{\text{lock}} \geq 1$ periods after any withdrawal request. Upon a fail result ($z_t^i = 1$), the protocol slashes a fraction $\lambda \in (0, 1]$ of the stake. In the event of slashing, the provider is barred from service provision until they replenish stake to $\underline{S}$. We write $S = \underline{S}$ throughout and assume stake is restored to $S$ at the start of each active period. Replenishment is immediate, with no constraints on liquidity or downtime. $S$ is then constant across active periods even when slashing is incurred.  
 
Each provider is assigned a \textit{reputation} $\rho_t^i$ from a finite set of reputation tiers $\mathcal{R}_n = \{G, P_1, \dots, P_n\}$, where $G$ denotes good standing and $P_1, \dots, P_n$ are probation tiers\footnote{Stake replenishment is also required for access to probation rewards. No rewards are given for undercollateralized providers in any case.}. The parameter $n$ is a protocol design choice representing the number of passed public outcomes required for reinstatement to good standing. Given the public outcome, reputation evolves deterministically by a transition function $T: \mathcal{R}_n \times \mathcal{Z} \to \mathcal{R}_n$ defined below: 

\begin{equation}
    T(\rho, z) = 
    \begin{cases} 
    G & \text{if } \rho = G \text{ and } z = 0  \\ 
    P_{j+1} & \text{if } \rho = P_j, \: j < n, \text{ and } z = 0\\
    G & \text{if } \rho = P_n \text{ and } z = 0\\
    P_1 & \text{if } z = 1 
    \end{cases}
    \label{BinaryTransitionFn}
\end{equation}
The reward function (defined below) assigns $r(G) = r_G$ and $r(P_j) = r_j$ for $j = 1, \dots, n$ with $r_1 \leq r_2 \leq \dots \leq r_n \leq r_G$. Protocol rewards depend only on the reputation tier, not on the service state or action. Provider $i$'s public history at time $t$ is then $h_t^i = (\rho_1^i, z_1^i, \dots, \rho_{t-1}^i, z_{t-1}^i, \rho_t^i)$.

\subsubsection{Timing and the Committed Protocol}\label{Timing}

Under a fixed protocol $\Pi$, play in each period proceeds as follows. The tier $\rho_t$ is public; the private state $x_t$ is drawn and the provider chooses $a_t$. The audit layer observes signals and elicits reports. Conditional on truthfulness of the audit layer, the protocol aggregates reports into $z_t$. Rewards $r(\rho_t)$ are paid, and slashing and probation transitions execute on fail outcomes. Exit and re-entry decisions are made at the period boundary before $x_{t+1}$.

\subsection{Payoffs}\label{Payoffs}

Given our players and information structures, we now model stage payoffs. Let the per-period service reward be  $r: \mathcal{R}_n \to \R_+$ and the cost\footnote{Note that the cost function implies that compliance \textit{must} incur a positive cost, a mild assumption in most DePIN protocols.} of action $a$ in state $x$ be $c: \mathcal{A} \times \mathcal{X} \to \mathbb{R}_+$ where $c(a^*, x) > 0 \;\: \forall \: x$. Let $b: \mathcal{A} \times \mathcal{X} \to \mathbb{R}_{+}$ denote the gross private benefits dependent on the action, which captures revenue a deviation may generate beyond protocol rewards. Let $\kappa > 0$ be the per-period opportunity cost of locked capital. 
The stage payoff for provider $i$ in period $t$ is:
\[u_t^i = b(a_t^i, x_t^i) - c(a_t^i, x_t^i) + r(\rho_t^i) - \lambda S \cdot \mathds{1}[z_t^i = 1] - \kappa S]\]
where $\mathds{1}[z_t^i = 1]$ represents the indicator function and $\lambda S$ is the slashing penalty. The terms $r(\rho)$ and $\kappa S$ are tier-dependent and constant respectively. $\kappa S$ is constant at every tier and action and cancels from all incentive comparisons. We supress it in the value function, reintroducing it in the participation analysis of section \ref{ParticipationConstraint}. In the repeated game, provider $i$ maximizes:
\[\mathbb{E} \left[\sum_{t=1}^\infty \delta^{t-1}u_t^i\right]\]
We further formalize this object as the continuation value function in Section~4. 

\begin{definition}[Compliance] \label{def:compliance}
    Fix a committed protocol $\Pi$. The protocol implements compliance if, after every public history summarized by $\rho$ and private state $x$, the provider weakly prefers compliance $a^*$ to every deviation given the continuation values induced by $\Pi$. 
\end{definition}
The provider's incentive to deviate depends on both the current state and the nature of the deviation. Following the moral hazard literature \cite{holmstrom1979}, the key quantity is the total private gain from deviating over complying. For a deviation $d \in \mathcal{A} \setminus \{a^*\}$, we define the \textit{state-conditional gain}
\[\Delta_d(x) = \bigl[b(d,x) - c(d,x)\bigr] - \bigl[b(a^*,x) - c(a^*,x)\bigr],\]
which reduces to the cost saving $c(a^*, x) - c(d, x)$ when the deviation generates no new revenue. Gains whose magnitude depends on the provider's own reputation tier are outside the model. Since $\mathcal{A}$ is an arbitrary finite set, a deviation bundled with a particular costly signal manipulation is itself an element of $\mathcal{A}$; the costless-manipulation case is recovered when all such bundles share a common cost. The set of \textit{state-relevant profitable deviations} is:
\begin{equation}
    \mathcal{D}(x) = \{d \in \mathcal{A} \setminus \{a^*\} \mid \Delta_d(x) > 0\}
    \label{eq:ProfitableDeviations}
\end{equation}
Crucially, $\mathcal{D}(x)$ depends on the state: a deviation profitable in one state may be unprofitable in another, so incentive compatibility must hold pairwise across every $(d, x)$ with $\Delta_d(x) > 0$. We then require the following assumption:
\begin{assumption}[Non-Trivial Incentive Problem]\label{NontrivialAssumption}
    There exist $x \in \operatorname{supp}(\mu)$ and $d \in \mathcal{A} \setminus \{a^*\}$ with $\Delta_d(x) > 0$. 
\end{assumption}
\section{State-Conditional Separation}\label{StateConditionalSeparation}

\subsection{Public Monitoring and State-Conditional Detection}\label{FullRank}

The protocol is separated into two layers: the audit layer, which observes provider signals and delivers reports conditional on the auditing mechanism being incentive-compatible, and the enforcement layer, which compresses these reports into a public outcome and applies the reward, slashing, and reputation logic. From the provider's perspective, only the enforcement layer affects dynamic incentives: after choosing an action, the provider faces a distribution over public outcomes, which determines current penalties and future reputation. 

Our baseline case is binary public outcomes. The coarseness of the pass and fail outcomes should be seen as a deliberate choice. On-chain logic will typically not discern between different manipulation types or deviation types, which means that it cannot condition punishments on the type of deviation. This is loosely related to the \textit{full-rank condition} of Fudenberg et al. \cite{flm1994}, in which it is required that outcome distributions induced by different player actions be linearly independent. We do not impose this assumption because it is highly unrealistic for real-world DePIN protocols, as even the most ambitious protocol designer will fail to enumerate the set of actions available to providers when setting corresponding punishments. 

The relevant monitoring question is therefore not whether the audit layer identifies an action, but whether each profitable deviation makes punishment sufficiently more likely than it is under compliance. If a deviation saves cost without increasing the probability of a fail outcome, then the protocol's binary punishment is ineffective. If in the worst case the deviation lowers the fail probability compared to compliance, then the provider both saves cost and becomes less likely to be punished than compliant providers. This motivates our concept of a \textit{detection gap}. For each private state $x$, compliance induces a fail probability $\eta(x)$. A profitable deviation $d \in \mathcal{D}(x)$, paired with the provider's best feasible manipulation, creates a worst-case probability. We formalize this in the sections below.

\subsection{Separation and Detection}\label{SeparationDetection}
We take the monitoring technology as a primitive, with a more rigorous derivation in the appendix section \ref{app: foundationDetectionPrimitives}. For each $x \in \operatorname{supp}(\mu)$ and each deviation $d$, the audit process fails a compliant provider with probability\footnote{We maintain that audits are neither perfect or completely uninformative under compliance.} $\eta(x) \in (0, 1)$ and a deviating provider with probability $P_d(x)$, where $P_d(x)$ is the fail probability under a provider's best available manipulation of signals. 


\begin{definition}[State-Conditional Audit Separation]\label{StateConditionalAuditSeparation}
Define the detection gap $\varepsilon_d(x) := P_d(x) - \eta(x)$ for each $x \in \operatorname{supp}(x)$ and $d \in \mathcal{A} \setminus \{a^*\}$. The audit mechanism satisfies \textit{state-conditional audit separation} if $\varepsilon_d(x) > 0$ for every $x \in \operatorname{supp}(\mu)$ and $d \in \mathcal{D}(x)$.
\end{definition}

Separation requires that deviation increases fail probability even under optimal manipulation. This rules out deviations where the provider simultaneously saves costs and reduces their measured failure rate. When such deviations exist, compliance cannot be implemented under binary public outcomes under any reward or collateral scheme.

We additionally assume that no deviation is both costlier than compliance and less likely to fail.
\begin{assumption}[Compliance Signal-Optimality]\label{OptimalComplianceAssumption}
For every $x \in \operatorname{supp}(\mu)$ and every $d \in \mathcal{A} \setminus \{a^*\}$ with $\Delta_d(x) \leq 0$, we have $\varepsilon_d(x) \geq 0$: no action weakly costlier than compliance
attains a strictly lower fail probability, even under optimal manipulation.
\end{assumption}
Absent this assumption, large punishments can make costlier yet safer actions attractive, which mirrors the over-compliance phenomenon of Craswell and Calfee \cite{craswellcalfee86}. The incentive condition becomes two-sided: the maximal tier enforcement $\lambda S + \delta \max_{\rho} \Phi_{\rho}$ would also be bounded above by $\min \Delta_d(x)/\varepsilon_d(x)$ over such actions.

\begin{proposition}[Impossibility Without Separation] \label{ImpossibilityWithoutSeparation}
Consider committed protocols in which the rewards, slashing, and continuation values condition only on the binary public outcome $z \in \{0, 1\}$, where $z = 1$ weakly reduces the provider's current or continuation payoff relative to $z=0$. If there exists $x \in \operatorname{supp} (\mu)$ and $d \in \mathcal{D}(x)$ such that $\varepsilon_d(x) \leq 0$, then no such binary public-outcome protocol implements compliance against deviation $d$ in state $x$. 

\end{proposition}

Note that this proposition is specific to mechanisms conditioning on binary outcomes. Richer structures of $\mathcal{Z}$ or other methods of out-of-protocol verification may deter deviations that are otherwise undetectable under binary aggregation.

\subsection{Deterrence Ratio}\label{DeterrenceRatio}

The separation condition gives us a detection gap $\varepsilon_d(x)$ for each state--deviation pair. The relevant incentive question is whether this gap is large enough relative to the temptation gap $\Delta_d(x)$.
\begin{definition}[Deterrence Ratios] \label{def:deterrence-ratios}
    For each $x \in \operatorname{supp}(\mu)$ and $d \in \mathcal{D}(x)$, the deterrence ratio is $\gamma_d(x) := \Delta_d(x)/\varepsilon_d(x)$, and the global deterrence ratio is
    \[\Gamma := \max_{x \in \operatorname{supp}(\mu),\, d \in \mathcal{D}(x)} \gamma_d(x)\]
\end{definition}
The global ratio captures the binding state--deviation pair: the combination hardest to deter. 

\section{Incentive Constraints and Dynamic Implementation}\label{IncentiveConstraintsAndResults}
Section 3 established that binary outcomes carry information in the form of detection gaps for state--deviation pairs and compliance, summarized in the global deterrence ratio $\Gamma$. We now analyze whether this information, combined with protocol enforcement tools of slashing and reputation, is sufficient to make compliance incentive compatible in general. We show that the answer reduces to a single inequality. 

\subsection{Main Incentive Inequality}
Fix a period $t$ and suppose state-conditional audit separation holds for the audit mechanism. Let $V: \mathcal{R}_n \to \mathbb{R}$ denote the continuation value under compliant play where $V(\rho)$ is the ex ante expected discounted payoff from the current period onward for a compliant provider in tier $\rho$. Define the \textit{continuation punishment} $\Phi_{\rho} = V(T(\rho, 0)) - V(T(\rho, 1))$ as the value lost from a fail outcome relative to a pass.
\begin{lemma}[Main Incentive Inequality]\label{MainIncentiveLemma}
    Fix a committed protocol $\Pi$ and continuation values $V$ and suppose a fail outcome weakly reduces the provider's total payoff at every tier. Compliance is sequentially incentive compatible against all deviations $d \in \mathcal{D}(x)$ if and only if, after every reputation tier $\rho$ and private state $x$:
    \begin{equation}
        \lambda S + \delta \min_{\rho \in \mathcal{R}_n}\Phi_{\rho} \geq \Gamma
    \label{eq:MainInequality}
    \end{equation}
\end{lemma}
\begin{proof}
    Consider a provider in state $x$ contemplating deviation $d \in \mathcal{D}(x)$. Under compliance, the provider pays cost $c(a^*, x)$ and faces false-positive detection probability $\eta(x)$. Under deviation $d$, the provider pays $c(d, x)$ and faces fail probability greater than or equal to $\eta(x) + \varepsilon_d(x)$. The expected payoff loss from deviating rather than complying is: 
    \begin{equation}
        [c(d, x) - c(a^*, x)] + \varepsilon_d(x) \cdot [\lambda S + \delta\Phi_{\rho}]
        = -\Delta_d(x) + \varepsilon_d(x) \cdot [\lambda S + \delta \Phi_{\rho}]
        \label{eq:ExpectedPayoffDifference}
    \end{equation}
    Deviation is then deterred if and only if \eqref{eq:ExpectedPayoffDifference} is nonnegative. Rearranging and dividing by $\varepsilon_d(x)$: 
    \[\lambda S + \delta \Phi_{\rho} \geq \frac{\Delta_d(x)}{\varepsilon_d(x)} \quad \forall x \in \operatorname{supp}(\mu), \:\: \forall d \in \mathcal{D}(x)\]
    Since the left side does not depend on $d$ or $x$, the pointwise condition holds for binding pairs if and only if it holds for the maximizing pair. Since we require this to hold at every tier, the binding constraint is the tier with the smallest $\Phi_{\rho}$. Taking the minimum over tiers and maximum over state--deviation pairs:
    \[\lambda S + \delta \min_{\rho \in \mathcal{R}_n}\Phi_{\rho} \geq \Gamma\]
    \qed
    
\end{proof}

\subsection{The One-Shot Game}
In the one-shot game, providers have no reputation or future values. A provider in state $x$ chooses compliance $a^*$ or a deviation $d$ by comparing: 
\[[r - c(a^*, x) - \eta(x)\lambda S] - [r - c(d, x) - (\eta(x) + \varepsilon_d(x))\lambda S] = -\Delta_d(x) + \varepsilon_d(x) \lambda S\]

Setting this difference nonnegative and rearranging gives $\lambda S \geq \Delta_d(x)/\varepsilon_d(x) = \gamma_d(x)$. Taking the max over all binding pairs, compliance is preferred if and only if $\lambda S \geq \Gamma$, which is exactly \eqref{eq:MainInequality} with $\delta \Phi = 0$, or the staking-dominated regime. For $d \notin \mathcal{D}(x)$, Assumption~\ref{OptimalComplianceAssumption} gives $\varepsilon_d(x) \geq 0$, so the difference above is nonnegative at any stake level and such deviations are never preferred.

The one-shot designer thus relies on collateral alone, which grows massively when deviations are subtle and $\Gamma$ is large; this motivates the repeated game, where reputation can substitute for much or all of the economic cost of required collateral.

\subsection{Continuation Punishment}\label{ContinuationPunishments}
Our main incentive inequality \eqref{eq:MainInequality} depends on $\min_{\rho} \Phi_{\rho}$, which is determined by the reputation structure and reward schedule. We now derive this quantity from the primitives of our model. 

Let $\bar{c} = \sum_x \mu(x) \cdot c(a^*, x)$ denote the expected cost of compliance and $\eta = \sum_x \mu(x)\eta(x)$ denote the average false positive rate. Under compliant play, the continuation values satisfy the following Bellman system: 
\begin{equation}
    \begin{aligned}
    &V_G = r_G - \bar{c} - \eta \lambda S + \delta[(1-\eta)V_G + \eta V_{P_1}]& \\
    &V_{P_j} = r_j - \bar{c} - \eta \lambda S + \delta[(1 - \eta)V_{P_{j+1}} + \eta V_{P_1}]\quad \text{ for } \: j < n& \\
    &V_{P_n} = r_n - \bar{c} - \eta \lambda S + \delta [(1 - \eta)V_G + \eta V_{P_1}]&
    \end{aligned}
    \label{eq:BellmanSystem}
\end{equation}
These equations are computed under compliant play, where the provider chooses $a^*$ in every state. Since the provider has not observed future $x$ when the continuation value is evaluated, and $x_t$ is i.i.d., the expected per-period cost and false positive rates are $\bar{c}$ and $\eta$ respectively. 

\begin{lemma}[Existence and Uniqueness of Continuation Values] \label{EUContinuationValuesLemma}
The Bellman system \eqref{eq:BellmanSystem} has a unique solution $\{V(\rho)\}_{\rho \in \mathcal{R}_n}$ with all finite values.
\end{lemma}

\begin{proposition}[Continuation Punishment with $n$ Probation Tiers] \label{ContinuationPunishmentProp}
In the $n$-tier probation system, suppose the tier-dependent rewards $r_G, r_1, \dots, r_n$ are monotone. The binding continuation punishment is:
\begin{equation}
    \min_{\rho \in \mathcal{R}_n} \Phi_{\rho} = \Phi_{P_1} = \sum_{j=1}^{n-1} \alpha^{j-1}(r_{j+1} - r_j) + \alpha^{n-1}(r_G - r_n)
    \label{eq:BindingContinuationPunishment}
\end{equation}
where $\alpha = \delta(1 - \eta)$. Moreover, monotone rewards imply $V_{P_1} \leq V_{P_2} \leq \dots \leq V_{P_n} \leq V_G$.
\end{proposition}

\subsection{Participation Constraints and Deterring Sybil Re-entry}\label{ParticipationConstraint}
A provider who exits the protocol earns a per-period \textit{outside option} $\bar{u} \geq 0$, representing the return from deploying their capital and hardware in alternative use. This value may be small in protocols with specialized equipment but will remain positive with our assumption that locked capital can always be deployed elsewhere. 

Under compliance, a provider incurs expected false positive slashing cost of $\eta \lambda S$ per period and opportunity cost $\kappa S$ for locked capital, where $\kappa > 0$ is the per-period cost of locked capital. The participation constraint, or the condition under which a provider chooses to remain in the protocol, must also hold at the binding probation tier to give the threat of demotion credible force: $r_1 - \bar{c} - \eta \lambda S - \kappa S \geq \bar{u}$. Rearranging, the minimum probation reward is then: 
\begin{equation}
    r_1 \geq \bar{u} + \bar{c} + S(\eta \lambda + \kappa)
    \label{eq:MinProbationReward}
\end{equation}
Equation \eqref{eq:MinProbationReward} is a sufficient flow condition.\footnote{$V$ omits the per-period capital cost $\kappa S$. Defining $W(\rho) = V(\rho) - \kappa S/(1 - \delta)$, the dynamic participation condition at tier $\rho$ is $W(\rho) \geq \bar{u}/(1 - \delta)$.}

\begin{assumption}[Cost of Re-entry] \label{CostofReEntry}
Let $U_{\text{lock}} = \bar{u} \cdot \delta(1 - \delta^{\Lambda_{\text{lock}}})/ (1 - \delta)$ be the present value of forgone outside options during the lockup period. The cost of re-entry must satisfy:
\[\kappa_{\text{entry}} + \lambda S + U_{\text{lock}} \geq \delta(V_G - V_{P_1})\]
where $\kappa_{\text{entry}} \geq 0$ is the cost of establishing a new identity, the price of a cheap pseudonym in the sense of \cite{friedmanresnick2001}\footnote{This assumption requires that each hardware node is limited to one on-chain identity, as creating two identities tied to the same hardware would allow trivial avoidance of probation. This assumption may be strong in certain DePIN verticals, but is typically handled with on-device onboarding keys or hardware attestation.}.
\end{assumption}

The left side is the total cost of re-entry: onboarding expenses, new additional capital requirements, and foregone outside option earnings during the lockup period while both stakes are committed. The right side represents the discounted benefit of escaping probation. Because $U_{\text{lock}}$ is increasing in $\Lambda_{\text{lock}}$ but bounded above by $\delta \bar{u}/(1 - \delta)$ and $\lambda \leq 1$, lockup and stake cannot always deter identity reset. Assumption \ref{CostofReEntry} is a feasibility condition, satisfiable whenever the continuation gain $\delta(V_G - V_{P_1})$ is below the maximal re-entry cost the protocol can impose. The key tradeoff is that increasing these requirements will likely raise barriers for honest new providers. We formalize this optimization problem in Section 5. 

\subsection{Compliance Incentive Compatibility} \label{MainEquilibriumTheorem}

The two terms in \eqref{eq:MainInequality} are the designer's two enforcement channels: the immediate slashing penalty $\lambda S$ and the discounted reputation punishment $\delta \min_{\rho} \Phi_{\rho}$. Our main result shows that when these channels jointly cover $\Gamma$ and providers find participation worthwhile, compliance is sequentially incentive compatible.

\begin{theorem}[Compliance Under Slashing and Reputation]
Fix a committed protocol $\Pi = (S, \lambda, r, T, \Lambda_{\text{lock}})$. Suppose Assumptions \ref{NontrivialAssumption} and \ref{OptimalComplianceAssumption} hold, the audit layer satisfies state-conditional separation, and participation and re-entry are satisfied. Then compliance is sequentially incentive compatible against every deviation after every public reputation tier and private state if and only if 

\[\lambda S+\delta \min_{\rho \in \mathcal{R}_n} \Phi_{\rho} \geq \Gamma\]

\end{theorem}

\begin{proof}
    By Lemma \ref{EUContinuationValuesLemma}, the Bellman system \eqref{eq:BellmanSystem} has a unique bounded solution $V$, the continuation values induced by the committed protocol under compliant play. Since the reputation transitions are deterministic, any public history is summarized by the reputation tier $\rho$. By Lemma~\ref{MainIncentiveLemma} applied at these continuation values, compliance is weakly preferred to every profitable deviation $d \in \mathcal{D}(x)$ after every tier and state if and only if the stated inequality holds. 
    
    We now verify deviations $d \notin \mathcal{D}(x)$. By Proposition \ref{ContinuationPunishmentProp}, monotone rewards imply $V_{P_1} \leq V_{P_2} \leq \dots \leq V_G$, which implies that $\Phi_{\rho} \geq 0$ for every $\rho$. Hence, $\lambda S + \delta \Phi_{\rho} \geq 0$. For any $d$ with $\Delta_d(x) \leq 0$, Assumption \ref{OptimalComplianceAssumption} gives $\varepsilon_d(x) \geq 0$, and the payoff change from such a deviation, $\Delta_d(x) - \varepsilon_d(x) \cdot [\lambda S + \delta \Phi_{\rho}]$, is non-positive at every tier and every enforcement level. Compliance is therefore weakly preferred to every deviation if and only if it is preferred on $\mathcal{D}(x)$. 

    To prove the biconditional result, we show that the one-shot deviation principle \cite{mailathsamuelson06} applies, since stage payoffs are bounded and $\delta < 1$. $(\rho, x)$ is the complete state relevant to payoffs, so the single-period comparisons above exhaust all deviation strategies in the game. Exiting the protocol and re-entry are ruled out by participation in the flow condition of $\eqref{eq:MinProbationReward}$ and Assumption \ref{CostofReEntry} respectively. 
    \qed
\end{proof}

\section{Protocol Design}
Theorem 1 involves parameters the designer controls ($S$, $\lambda$, $n$, $r_i$) and primitives the service and state environments fix ($\Gamma$, $\eta$, $\bar{c}$, and $\bar{u}$). The designer's problem is to choose the controllable parameters to satisfy all constraints given the primitives, at minimum cost, which we formalize in this section. We solve the closed form of this optimization and show how it specializes across three broad categories of audit technologies that generate our detection primitives.

\subsection{Audit Frequency and Redundancy}\label{AuditFrequencyandRedundancy}
Before stating the designer's problem formally, we characterize how audit frequency and redundancy affect the separation condition and hence the deterrence ratio $\Gamma$. Suppose the protocol audits each provider independently with probability $p \in (0, 1]$ per period. When no audit occurs, the public outcome defaults to pass ($z = 0$). The effective false positive rate and detection gap scale accordingly. This scaling assumes the provider does not observe whether the current period will be audited before choosing $a_t$.

\begin{lemma}[Audit Frequency]\label{AuditFrequencyLemma}
Under random auditing with probability $p$, the effective false positive rate and detection gap are: 
\[\eta_{\text{eff}}(x) = p \cdot \eta_1(x), \qquad \varepsilon_{\text{eff}, d}(x) = p \cdot \varepsilon_{1, d}(x)\]
where $\eta_1(x)$ and $\varepsilon_{1, d}(x)$ are the rates under certain auditing $p = 1$. The deterrence ratio under audit frequency $p$ is: 
\[\Gamma(p) = \max_{x, d} \frac{\Delta_d(x)}{p \cdot \varepsilon_{1, d}(x)} = \frac{\Gamma_1}{p}\]
where $\Gamma_1$ is the deterrence ratio under certain auditing.
\end{lemma}

This scaling is the standard logic of inspection games \cite{avenhaus02}: reducing audit frequency inflates the deterrence ratio and hence the required collateral. The minimum stake under audit frequency $p$ is $S^*(p) = \Gamma_1 / (p\lambda)$ in the staking-dominated regime. The designer trades off direct audit cost against collateral requirements. 

\subsection{Reward Design}

The structure of \eqref{eq:BindingContinuationPunishment} gives the protocol designer insight into how to optimally set rewards for their choice of $n$ tiers. The binding continuation punishment is expressed as a discounted sum of the reward increments along the ladder. Taking $\Delta_j = r_{j + 1} - r_j$ for $j < n$ and $\Delta_n = r_G - r_n$, it is clear that the increments telescope towards the total reward span $\sum_{j=1}^n \Delta_j = r_G - r_1$. The geometrically declining weights in this equation then imply front-loading as the optimal reward scheme. 

\begin{proposition}[Front-Loaded Rewards]\label{FrontLoadingProp}
Fix the floor $r_1$ and ceiling $r_G$. For any monotone schedule $r_1 \leq \dots \leq r_n \leq r_G$, the continuation punishment $\Phi_{P_{1}}$ is maximized by the front-loaded schedule $r_2 = r_3 = \dots = r_n = r_G$
where $\Phi_{P_{1}} = r_G - r_1$.
\end{proposition}

Front-loading dominates the natural alternative of uniform probation rewards: a provider at $P_1$ who merely advances to an identically paid $P_2$ gains almost nothing from passing an audit, so the incentive to comply during probation weakens with each additional tier.

\begin{remark}[Uniform Probation Rewards]
Under uniform probation rewards $r_j = r_P$ for all $j \leq n$, all intermediate increments vanish and $\Phi_{P_1} = \alpha^{n-1}(r_G - r_P)$, which is decreasing in $n$. Any probation length $n \geq 2$ with uniform rewards will be counterproductive at the bounding tier. Setting a uniformly low probation reward and a high good-tier reward is strictly dominated by a front-loaded schedule in which the steepest gradient is between $P_1$ and $P_2$.
\end{remark}

\begin{proposition} [Optimal Probation Depth] \label{OptimalProbationDepth}
For every feasible set of design choices $(n, p, S, \mathbf{r})$, the two-tier design $(1, p, S, (r_1, \, r_1 + \Phi_{P_1}(n, \mathbf{r})))$ is feasible with the same deterrence and weakly lower reward cost. $n=1$ then attains the optimum of the designer's problem. 
\end{proposition}

Deeper ladders with more reward tiers arise from constraints outside of our program, such as requiring probation tiers to pay below set caps. 

\subsection{The Protocol Designer's Problem}
We now formalize the optimization problem from the perspective of the protocol designer. We fix $\lambda = 1$ throughout this section. Our incentive constraint depends on slashing only through the product $\lambda S$, while the cost of capital $\kappa S$ depends on $S$ alone. Thus the designer always prefers maximal slashing with the smallest stake meeting the constraint. 

Before stating the problem formally, we must derive the steady-state cost of operating the reputation system. Under compliance, the reputation tier evolves as a Markov chain on $\mathcal{R}_n$: a fail with probability $\eta$ sends the provider to state $P_1$, and a pass with probability $1 - \eta$ sends the provider one tier toward $G$. The chain is both irreducible and aperiodic, so it admits a unique stationary distribution. 

\begin{proposition}[Steady-State Distribution] \label{SteadyStateProp}
    Under compliance, the stationary distribution over tiers is:
\[\pi_G = (1 - \eta)^n, \qquad \pi_{P_j} = \eta(1 - \eta)^{j-1} \quad \text{for } j = 1, \dots, n\]
\end{proposition}

A provider spends $(1 - \eta)^n$ of time in good standing, which decreases exponentially in probation depth. The expected per-period reward paid to a compliant provider in steady-state is then: 
\[\bar{R}(n, \mathbf{r}) = (1 - \eta)^n r_G + \eta \sum_{j = 1}^n (1 - \eta)^{j -1} r_j\]

Under audit frequency $p$, the effective false positive rate is then $\eta = p \eta_1$ and the deterrence ratio is $\Gamma(p) = \Gamma_1/p$ by Lemma \ref{AuditFrequencyLemma}. The designer then solves:

\begin{equation}
\begin{aligned}
\min_{n,\, p,\, S,\, \mathbf{r}} \quad & \bar{R}(n, \mathbf{r}) 
   + p \cdot c_{\text{audit}} \\
\text{subject to} \quad 
& S + \delta\, \Phi_{P_1}(n, \mathbf{r}) \geq \Gamma(p) 
   && \text{(IC)} \\
& r_1 \geq \bar{u} + \bar{c} + S(\eta + \kappa) 
   && \text{(Participation)} \\
& r_1 \leq r_2 \leq \dots \leq r_n \leq r_G 
   && \text{(Monotonicity)} \\
& p \in (0, 1],\quad S \in [0, \bar{S}],\quad n \in \mathbb{Z}_+
\end{aligned}
\label{eq:DesignersProblem}
\end{equation}

where $\Phi_{P_1}(n, \mathbf{r})$ is given by \eqref{eq:BindingContinuationPunishment} with $\alpha = \delta(1 - \eta)$. We omit the re-entry condition (Assumption \ref{CostofReEntry}) from the program: it is a maintained feasibility condition on the environment (Section~\ref{ParticipationConstraint}), and since $\Lambda_{\text{lock}}$ contributes no cost to the objective, the designer sets it at the maximum consistent with honest entry after solving the remaining parameters. 

The program is tractable and directly computable. For fixed probation depth $n$ and audit frequency $p$, the effective rates $\eta$ and $\Gamma(p)$ are constants and the remaining problem in $(S, \mathbf{r})$ is a linear program: the objective $\bar{R}$ is linear in rewards, $\Phi_{P_1}$ is linear in the increments between rewards, and the participation and monotonicity constraints are linear in $(S, \mathbf{r})$. By Proposition \ref{OptimalProbationDepth} the search over $n$ collapses to $n=1$, yielding a linear program over a one dimensional grid. 

\begin{proposition}[Optimum Binding Constraints] \label{BindingConstraintsProp}
At any solution of \eqref{eq:DesignersProblem}, the IC and participation constraints bind.
\end{proposition}

At the optimum, the designer then equates the marginal cost of collateral enforcement with the marginal cost of the reward spread. The designer chooses audit frequency to balance direct audit costs against the collateral and rewards required to cover $\Gamma(p)$.

\section{Discussion: Audit Classes and Enforcement Regimes}
The primitives of detection are generated by the specific technology and process of the audit layer. With our framework, we can classify three types of audit technologies based on how the detection gap $\varepsilon_d$ arises. Each class is in a regime of the main inequality \eqref{eq:MainInequality}, separated by which enforcement lever or combination of levers covers $\Gamma$. In \textbf{Staking-dominated} regimes, collateral alone suffices. \textbf{Mixed} regimes leave some residual to reputation spreads. The \textbf{Reputation-dominated} regime covers cases in which the required collateral is infeasible. This gives protocol designers a clear procedure when estimating their primitives and iterating on the audit layer: identify the binding deviation, determine how the detection gap arises, and invest in the lever that the corresponding regime exposes. 

\subsection{Attestation-Backed Networks}
In DePINs where work is verified through cryptographic attestations (i.e. proofs of replication, storage, or spacetime), the detection gap is near-maximal and $\Gamma$ is small. This is due to the nature of cryptographic audits, where discarding data is detectable within the small soundness error \cite{fisch2019}. Modest collateral requirements can then cover this cost, which induces the staking-dominated regime. As a result, the false positive rate $\eta$ is the binding variable, priced through the participation floor $\eta \lambda S$. Attestation-backed DePINs then concentrate effort on designing gradual fault fees and forgiveness given honest histories. 

\subsection{Physically Bounded Networks}
For networks that rely on physically located witnessing, like proof-of-coverage or location proving, exact positions cannot be verified against collusion \cite{chandran09}. $\varepsilon_d$ is then bounded by the graph structure of the witness set, where a deviating region is caught when challenges cross its cut boundary. $\Gamma_1$ is then moderate and induces the mixed regime. For illustration, take $\Gamma_1 = 20$, $\eta_1 = 0.05$, $\bar{c} = 10$, $\bar{u} = 2$, $\kappa = 0.1$, $\delta = 0.95$, $\bar{S} = 10$ in cost units per epoch with $p=1$. The binding constraints give $d\bar{R}/dS = (\eta + \kappa) - (1 - \eta)/\delta = -0.85 < 0$, so the optimum holds maximal stake with $r_1 = \bar{u} + \bar{c} + \bar{S}(\eta + \kappa) = 13.5$ and reputation spread $(\Gamma_1 - \lambda\bar{S}) /\delta \approx 10.53$, which covers the deterrence that collateral cannot. The direct way to improve this is to lower $\Gamma_1$ by densifying or randomizing witness assignment, raising the probability that the challenges cross the deviating set. 

\subsection{Statistically Bounded Networks}
For networks that audit by sampling outputs of work, like inference workloads or compute workloads in general, enforcement is hardest to design around. The detection gap is the distinguishability of a cheap substitute's output from honest work under the samples the audit requests. More sophisticated and capable spoofing methods drive $\varepsilon_d$ to zero and $\Gamma_1$ to infinity, which means that the separation condition is the most critical issue to address. Even in cases when separation barely holds, the required collateral will often be infeasible to suffice alone and reputation must be leaned on heavily in the reputation-dominated regime. This is further weakened when provider hardware is liquid and participation is high. Protocols in this category must first concentrate effort on verifiable inference schemes, an area of active research \cite{toploc2025}. 

\bibliographystyle{splncs04}
\bibliography{refs}

\appendix
\section{Omitted Proofs}\label{app:proofs}

\noindent\textbf{Proof of Lemma~\ref{EUContinuationValuesLemma}.}
The Bellman system \eqref{eq:BellmanSystem} defines an operator $\mathcal{T}: \R^{|\mathcal{R}_n|} \to \R^{|\mathcal{R}_n|}$. For any two candidate value vectors $V, V'$:
\[|(\mathcal{T}V)_\rho - (\mathcal{T}V')_\rho| = \delta|(1-\eta)(V_{T(\rho,0)} - V'_{T(\rho,0)}) + \eta(V_{P_1} - V'_{P_1})| \leq \delta\|V - V'\|_\infty\]
Thus $\mathcal{T}$ is a contraction with modulus $\delta < 1$ under the sup norm. The result follows from Banach's fixed point theorem. \qed

\medskip
\noindent\textbf{Proof of Proposition~\ref{ContinuationPunishmentProp}.}
Since $T(\rho, 1) = P_1$ for all $\rho$, we have $\Phi_{\rho} = V(T(\rho, 0)) - V_{P_1}$. Monotone rewards and the Bellman system \eqref{eq:BellmanSystem} imply $V_{P_1} \leq V_{P_2} \leq \dots \leq V_{P_n} \leq V_G$, so $\Phi_{\rho}$ is minimized when $T(\rho, 0)$ is smallest. For $n \geq 2$, $T(P_1, 0) = P_2$ and the minimum is $\Phi_{P_1} = V_{P_2} - V_{P_1}$; for $n = 1$, $T(P_1, 0) = G$ and $\Phi_{P_1} = V_G - V_{P_1}$. Evaluating the closed form below covers both cases.

To compute $\Phi_{P_1}$, define $D_j = V_{P_{j+1}} - V_{P_j}$ for $j < n$ and $D_n = V_G - V_{P_n}$. Subtracting the two Bellman equations yields $D_n = r_G - r_n$. For $j < n$: $D_j = (r_{j+1} - r_j) + \alpha D_{j+1}$. Unrolling the recursion from $j=1$ gives the result. \qed

\medskip
\noindent\textbf{Proof of Lemma~\ref{AuditFrequencyLemma}.}
With probability $1 - p$, no audit occurs and $z = 0$ regardless of the provider's action. With probability $p$, the audit proceeds with fail probability $\eta_1(x)$ under compliance or at least $\eta_1(x) + \varepsilon_{1, d}(x)$ under deviation. The unconditional fail probability is therefore $p \cdot \eta_1(x)$ under compliance and at least $p \cdot (\eta_1(x) + \varepsilon_{1, d}(x))$ under deviation. The detection gap is $p \cdot \varepsilon_{1, d}(x)$, and since $\Delta_d(x)$ is unaffected by audit frequency, the deterrence ratio scales as $\Gamma_1/p$. \qed

\medskip
\noindent\textbf{Proof of Proposition~\ref{FrontLoadingProp}.}
Each increment $\Delta_j$ enters \eqref{eq:BindingContinuationPunishment} with weight $\alpha^{j-1}$, which is strictly decreasing in $j$ since $\alpha \in (0,1)$. Subject to $\sum_{j=1}^n \Delta_j = r_G - r_1$ and $\Delta_j \geq 0$ (monotonicity), the weighted sum is maximized by placing the entire span on the largest weight: $\Delta_1 = r_G - r_1$ and $\Delta_j = 0$ for $j \geq 2$, i.e., $r_2 = \dots = r_n = r_G$, giving $\Phi_{P_1} = r_G - r_1$. \qed

\medskip
\noindent \textbf{Proof of Proposition \ref{OptimalProbationDepth}}
With increments $\Delta_j$, Proposition \ref{SteadyStateProp} gives $\bar{R} = r_1 + \sum_{j=1}^{n} (1 - \eta)^j \Delta_j$, and $\Phi_{P_1} = \sum_{j=1}^{n} \alpha^{j-1} \Delta_j$. The two-tier case has the same $S$, $r_1$, and continuation punishment, so the IC and participation constraints are identical. Reward cost $r_1 + (1 - \eta)\Phi_{P_1} = r_1 + \sum_j \delta^{j-1}(1 - \eta)^j \Delta_j \leq \bar{R}$ in every term since $\delta < 1$. \qed

\medskip
\noindent\textbf{Proof of Proposition~\ref{SteadyStateProp}.}
The balance equations give $\pi_{P_1} = \eta(\pi_G + \sum_{j=1}^n \pi_{P_j}) = \eta$ and $\pi_{P_j} = (1-\eta)\pi_{P_{j-1}}$ for $j \geq 2$, so $\pi_{P_j} = \eta(1-\eta)^{j-1}$. Then $\pi_G = 1 - \eta\sum_{j=1}^n (1-\eta)^{j-1} = 1 - (1-(1-\eta)^n) = (1-\eta)^n$. \qed

\medskip
\noindent\textbf{Proof of Proposition~\ref{BindingConstraintsProp}.}
By way of contradiction, suppose that IC is slack at some solution. Then $r_G$ can be reduced by some small increment without violating constraints, strictly reducing $\bar{R}$, which contradicts optimality. Suppose similarly that participation is slack at some solution. Then, $r_1$ can be reduced, strictly reducing $\bar{R}$, a contradiction. \qed

\section{Theoretical Foundation of the Detection Primitives} \label{app: foundationDetectionPrimitives}
\subsection{Signal Space and Manipulation Correspondence}

Let $\mathcal{S}$ denote a finite \textit{signal space}. The action and state jointly determine a distribution over the vector of private signals observed by the $k$ auditors assigned to provider $i$. Providers can strategically influence which distribution over signals is realized through a \textit{manipulation correspondence} $L$:
\begin{equation}
    L: \mathcal{X} \times \mathcal{A} \rightrightarrows \Delta(\mathcal{S}^k)
    \label{eq:ManipulationCorrespondence}
\end{equation}
where $L(x, a) \subseteq \Delta(\mathcal{S}^k)$ is the set of feasible joint distributions over the vector of $k$ auditor signals that provider $i$ can induce with action $a$ and state $x$. This formalism is taken from Milionis et al., where the feasible set of application-specific signals (reports) is denoted similarly.
\begin{assumption}[Manipulation Structure] \label{ManipulationStructureAssumption}
For each $x \in \mathcal{X}$, let $F_{a^*}(\cdot \mid x) \in \Delta(\mathcal{S}^k)$
denote the joint distribution over the vector of $k$ auditor signals under
compliance.
    \begin{enumerate}[label=(\alph*)]
        \item Under compliance, the signal distribution is a singleton:
        $L(x, a^*) = \{F_{a^*}(\cdot \mid x)\}$\footnote{Assumption 4(a) can be relaxed to let compliant providers optimize their signal presentation; separation (Definition~\ref{StateConditionalAuditSeparation}) then requires deviators to be caught more often even when both sides optimize, with the theory otherwise mildly changed. We adopt the singleton case throughout.}.
        \item For each $(x, a) \in \mathcal{X} \times \mathcal{A}$, the set
        $L(x, a)$ is nonempty, compact, and convex in $\Delta(\mathcal{S}^k)$.
        \item $F_{a^*}(\mathbf{s} \mid x) > 0$ for all $x \in \mathcal{X}$ and
        for all $\mathbf{s} \in \mathcal{S}^k$.
    \end{enumerate}
\end{assumption}
These assumptions are mild for real-world DePIN protocols: (a) compliant providers cannot manipulate their signal distribution; (b) optima over manipulations exist, and randomization over manipulations is absorbed without loss of generality; (c) the compliance distribution has full support, preventing trivial detection from signal vectors impossible under compliance. Under conditionally independent observation, $F_{a^*} = f_{a^*}^{\otimes k}$ for a per-auditor marginal and (c) follows from full support of the marginal; we do not impose product structure, since both the state and the manipulation correlate signals across auditors (e.g., one spoofed measurement broadcast to all witnesses).

\subsection{Audit Layer and Public Outcomes}\label{AuditorReports}
Auditors receive signals generated by the provider's action and state, and an incentive-compatible reporting mechanism, like the signal-recovery procedure for unverifiable information in the sense of \cite{milionis2025}, elicits reports from them. 

Conditional on the audit layer satisfying the truthfulness condition, the protocol aggregates reports and public randomness into a public outcome $z_t^i \in \mathcal{Z}$, where $\mathcal{Z}$ is a finite set with $|\mathcal{Z}| \geq 2$. The rest of the paper takes this induced public outcome mechanism as a given. In the baseline binary case $\mathcal{Z} = \{0, 1\}$, $z = 1$ denotes a 'fail' and $z=0$ denotes 'pass'. Formally, the aggregation rule is a mapping 
\begin{equation}
    \mathcal{G}: \mathcal{S}^k \times \Omega \to \mathcal{Z}
    \label{eq:AggregationRule}
\end{equation}
where $k$ is the number of auditor reports and $\Omega$ is a source of public randomness. With truthful reports $(s_1, \dots, s_k)$ and randomness $\omega$, the public outcome is $z = \mathcal{G}(s_1, \dots, s_k, \omega)$. Our analysis applies to the binary outcome structure $\mathcal{Z} = \{0, 1\}$.

For each state $x \in \mathcal{X}$, define the compliance-induced distribution over public outcomes:
\begin{definition}[State-Conditional Compliance Distribution]
For each $x \in \mathcal{X}$:
\[P_0^x(z) := \Pr\bigl(Z = z \mid a^*, x\bigr) \quad \text{for each } z \in \mathcal{Z}\]
where the signal vector is drawn as $\mathbf{s} = (s_1, \dots, s_k) \sim
F_{a^*}(\cdot \mid x)$ and $Z = \mathcal{G}(\mathbf{s}, \omega)$.
\end{definition}
Now, for each deviation $d \in \mathcal{A} \setminus \{a^* \}$, the provider selects
the manipulation $F_d \in L(x, d)$ that minimizes the probability of a fail
outcome. This worst-case scenario gives the following condition its adversarial
robustness.
\begin{definition}[State-Conditional Deviation Distribution]
For each $x \in \mathcal{X}$ and $d \in \mathcal{A} \setminus \{a^* \}$: \[P_d^x(1) \;=\; \inf_{F_d \in L(x,d)}\; \Pr\bigl(Z = 1 \mid F_d\bigr), \qquad P_d^x(0) \;=\; 1 - P_d^x(1)\]
where $\mathbf{s} \sim F_d$ and $Z = \mathcal{G}(\mathbf{s}, \omega)$.
\end{definition}

Note that the infimum is attainable as $L(x, d)$ is compact (Assumption~\ref{ManipulationStructureAssumption}(b)) and
$\Pr(Z = 1 \mid \cdot)$ is linear and continuous in the joint signal distribution. The separation condition is imposed after the audit layer. Source identifiability in \cite{milionis2025} concerns whether the mechanism can recover the information, while our $\varepsilon_d(x) > 0$ condition concerns whether the recovered information has enough statistical power to dynamically induce compliance. 

This framework recovers the primitives of section \ref{SeparationDetection} as $\eta(x) = P_0^x(1)$ and $P_d(x) = P_d^x(1)$.
\end{document}